\documentclass[11pt]{article}

\usepackage[toc,page]{appendix}
\usepackage[normalem]{ulem} 
\usepackage{amssymb} 
\usepackage{amsmath} 
\usepackage{amsthm} 
\usepackage{setspace} 
\usepackage{cite} 

\usepackage{tikz}
\usetikzlibrary{calc,positioning,shapes,shadows,arrows,fit}

\newtheorem{definition}{Definition} 
\newtheorem{corollary}{Corollary} 
 
\newtheorem{claim}{Claim} 
\newtheorem{lemma}{Lemma} 
\newtheorem{theorem}{Theorem}

\newtheorem{proposition}{Proposition}
\newtheorem{remark}{Remark}

\usepackage[hmargin=2.9cm,vmargin=2.9cm]{geometry} 

\newcommand{\F}{\mathbb{F}}

\newcommand{\M}{\mathcal{M}}

\renewcommand{\angle}[1]{\mathopen{\langle} #1\mathclose{\rangle}}

\DeclareMathOperator{\poly}{\mbox{\small\rm poly}}

\title{Randomized Polynomial Time Identity Testing for Noncommutative Circuits}

\author{V. Arvind\thanks{Institute of Mathematical Sciences, Chennai,
    India, \texttt{email: arvind@imsc.res.in}} \and Partha
  Mukhopadhyay\thanks{Chennai Mathematical Institute, Chennai, India,
    \texttt{email: partham@cmi.ac.in}} \and S. Raja\thanks{Institute of Mathematical Sciences, Chennai, India,
    \texttt{email: rajas@imsc.res.in}}}
 
\begin{document} 

\maketitle

\begin{abstract} 
In this paper we show that the black-box polynomial identity testing for
noncommutative polynomials $f\in\mathbb{F}\langle z_1,z_2,\cdots,z_n \rangle$ of degree $D$ and sparsity $t$, 
can be done in randomized $\poly(n,\log t,\log D)$ time. 
As a consequence, if the black-box contains a circuit $C$ of size $s$ computing $f\in\mathbb{F}\langle z_1,z_2,\cdots,z_n \rangle$ 
which has at most $t$ non-zero monomials, then the identity testing can be done by a randomized algorithm
with running time polynomial in $s$ and $n$ and $\log t$. This makes significant progress on a question
that has been open for over ten years.

The earlier result by Bogdanov and Wee \cite{BW05}, using the
classical Amitsur-Levitski theorem, gives a randomized polynomial-time
algorithm only for circuits of polynomially bounded syntactic degree.
In our result, we place no restriction on the degree of the circuit.

Our algorithm is based on automata-theoretic ideas introduced in
\cite{AMS08,AM08}. In those papers, the main idea was to construct
deterministic finite automata that isolate a single monomial from the
set of nonzero monomials of a polynomial $f$ in
$\mathbb{F}\langle z_1,z_2,\cdots,z_n \rangle$. In the present paper, since we need to
deal with exponential degree monomials, we carry out a different kind
of monomial isolation using nondeterministic automata.
\end{abstract}

\section{Introduction}

Noncommutative computation, introduced in complexity theory by Hyafil
\cite{Hya77} and Nisan \cite{N91}, is a central field of algebraic
complexity theory. The main algebraic structure of interest is the
free noncommutative ring $\F\angle{Z}$ over a field $\F$, where
$Z=\{z_1,z_2,\cdots,z_n\}$, $z_i, 1\le i\le n$ are free noncommuting
variables.

One of the main problems in the subject is noncommutative Polynomial
Identity Testing. The problem can be stated as follows: 

Let $f\in \F \angle{Z}$ be a polynomial represented by a noncommutative
arithmetic circuit $C$.  The polynomial $f$ can be either given by a
black-box for $C$ (using which we can evaluate $C$ on matrices with
entries from $\F$ or an extension field), or the circuit $C$ may be
explicitly given. The algorithmic problem is to check if the
polynomial computed by $C$ is identically zero. 

We recall the formal definition of a noncommutative arithmetic
circuit.

\begin{definition}
A \emph{noncommutative arithmetic circuit} $C$ over a field $\F$ and
indeterminates $z_1,z_2,\cdots,z_n$ is a directed acyclic graph (DAG)
with each node of indegree zero labeled by a variable or a scalar
constant from $\F$: the indegree $0$ nodes are the input nodes of the
circuit. Each internal node of the DAG is of indegree two and is
labeled by either a $+$ or a $\times$ (indicating that it is a plus
gate or multiply gate, respectively). Furthermore, the two inputs to
each $\times$ gate are designated as left and right inputs which is
the order in which the gate multiplication is done. A gate of $C$ is
designated as \emph{output}.  Each internal gate computes a polynomial
(by adding or multiplying its input polynomials), where the polynomial
computed at an input node is just its label. The \emph{polynomial
  computed} by the circuit is the polynomial computed at its output
gate. An arithmetic circuit is a formula if the fan-out of every gate
is at most one.
\end{definition}

Notice that if the size of circuit $C$ is $s$ the degree of the
polynomial computed by $C$ can be $2^s$. In the earlier
result \footnote{ We also note here that Raz and Shpilka
  \cite{raz05PIT} gives a white-box deterministic polynomial-time
  identity test for noncommutative algebraic branching programs
  (ABPs). The result of Forbes-Shpilka \cite{FS13} and Agrawal et al.,
  \cite{AGKS15} gives a quasi-polynomial time black-box algorithm for
  small degree ABPs.} by Bogdanov and Wee \cite{BW05}, a randomized
polynomial-time algorithm was shown for the case when the degree of
the circuit $C$ is polynomially bounded in $s$ and $n$
\cite{BW05}. The idea of the algorithm is based on a classical result
of Amitsur-Levitski \cite{AL}.  We recall below that part of the
Amitsur-Levitski theorem that is directly relevant to polynomial
identity testing.

\begin{theorem}[Amitsur-Levitski Theorem]\label{thm_al}
For any field $\F$ (of size more than $2d-1$), a nonzero
noncommutative polynomial $P \in \F\angle{Z}$ of degree $2d-1$ cannot
be a polynomial identity for the matrix algebra $\mathbb{M}_d(\F)$.
I.e.\ $f$ does not vanish on all $d\times d$ matrices over $\F$.
\end{theorem}

Bogdanov and Wee's randomized PIT algorithm \cite{BW05} applies the
above theorem to obtain a randomized PIT as follows: Let
$C(z_1,z_2,\cdots,z_n)$ be a circuit of syntactic degree bounded by
$2d-1$. For each $i\in [n]$, substitute the variable $z_i$ by a
$d\times d$ matrix $M_i$ of commuting indeterminates.  More precisely,
the $(\ell,k)^{th}$ entry of $M_i$ is $z^{(i)}_{\ell,k}$ where $1\leq
\ell,k\leq d$. By Theorem \ref{thm_al}, the matrix $M_f = f(M_1, M_2,
\ldots, M_n)$ is not identically zero. Hence, in $M_f$ there is an
entry $(\ell',k')$ which has the \emph{commutative} nonzero polynomial
$g_{\ell',k'}$ over the variables $\{z^{(i)}_{\ell,k} : 1\leq i\leq n,
1\leq \ell,k\leq d\}$. Notice that the degree of the polynomial
$g_{\ell',k'}$ is at most $2d-1$. If we choose an extension field of
$\F$ of size at least $4d$, then we get a randomized polynomial
identity testing algorithm by the standard
Schwartz-Zippel-Lipton-DeMello Lemma \cite{Sch80,Zippel79,DL78}.

The problem with this approach for general noncommutative circuits
(whose degree can be $2^s$) is that the dimension of the matrices
grows linearly with the degree of the polynomial. Therefore, this
approach only yields a randomized exponential time algorithm for the
problem. It cannot yield an efficient algorithm for polynomial
identity testing for a general noncommutative circuit where the
syntactic degree could be exponential in the size of the
circuit. Finding an efficient randomized identity test for general
noncommutative circuits was a well-known open problem. In this work we
resolve it for the case when the polynomial is promised to have the number of non-zero monomials  
at most exponential in the input size.

\section{Main Result} 

The crux of our result is the following theorem that we show about
noncommutative identities which is of independent mathematical
interest. 

\begin{theorem}\label{thm-main}
Let $\F$ be a field of size more than $d$. Let $f\in \F\langle z_1,
z_2, \ldots, z_n\rangle$ be a nonzero polynomial of degree $d$ and
with $t$ nonzero monomials. Then $f$ cannot be a polynomial identity
for the matrix ring $\mathbb{M}_k(\F)$ for $k=\log t + 1$.
\end{theorem}

The randomized polynomial identity test for noncommutative arithmetic
circuits is an immediate corollary. To see this, suppose $C$ is a
noncommutative arithmetic circuit of size $s$ computing a polynomial with at most 
$t$ monomials. The degree of the
polynomial $f$ computed by the circuit is bounded by $2^s$.
Thus, if $f$ is not identically zero, by
Theorem~\ref{thm-main}, the polynomial $f$ does not vanish if we
substitute for each $z_i$, $(\log t+1)\times (\log t+1)$ matrices of
indeterminates (all distinct). Indeed, $f$ will evaluate to an
$(\log t+1)\times (\log t+1)$ matrix whose entries are polynomials in commuting
variables of degree at most $2^s$ \footnote{Note that in general a noncommutative circuit of size $s$ can compute a 
polynomial that can have $2^{2^s}$ monomials. For example the polynomial $f(x,y)=(x+y)^{2^s}$ has noncommutative circuit of size 
$O(s)$ but the number of monomials is $2^{2^s}$. Our algorithm can not handle such cases.}. For each entry of this matrix, we
can employ the standard Schwartz-Zippel-Lipton-DeMello
\cite{Sch80,Zippel79,DL78} lemma based algorithm for commutative
polynomials (by evaluating them over $\F$ or a suitable extension
field). This proves the main result of the paper.

\begin{corollary}
Let $C$ be a noncommutative circuit of size $s$ given as a black-box computing a polynomial 
$f\in \F\angle{z_1,z_2,\ldots,z_n}$ with sparsity $t$. Then there is a randomized algorithm to check whether 
$f$ is an identically zero polynomial and the algorithm runs in time $\poly(s,n,\log t)$.  
\end{corollary}

\begin{remark}
It is interesting to compare Theorem~\ref{thm-main} with the classical
Amitsur-Levitski theorem. Our result brings out the importance of the
number of monomials in a polynomial identity for $d\times d$ matrices.
It implies that any polynomial identity $f$ for $d\times d$ matrices
over a field $\F$ of size more than $\deg f$ must have more than
$2^{d-1}$ monomials.
\end{remark}

We first describe the basic steps required for the proof of
Theorem~\ref{thm-main}. Since we are working in the free
noncommutative ring $\F\langle z_1, z_2, \ldots, z_n\rangle$, notice
that monomials are free words over the alphabet
$\{z_1,z_2,\ldots,z_n\}$, and the polynomial $f$ is an $\F$-linear
combination of monomials.

\subsubsection*{Converting to a bivariate polynomial}

It is convenient to convert the given noncommutative polynomial into a
noncommutative polynomial in $\F\langle x_0,x_1\rangle$, where $x_0$
and $x_1$ are two noncommuting variables. Let 
\[
f = \sum_{i=1}^t c_i w_i
\]
with $c_i \in \F$, where $w_i$ are the nonzero monomials (over
$\{z_1,z_2,\ldots,z_n\}$) of $f$. We use the bivariate substitution
$\forall i\in[n] : z_i\rightarrow x_0 x_1^i x_0$ to encode the words
over two variables $x_0, x_1$.  By abuse of notation, we write the
resulting polynomial as $f(x_0, x_1) \in\F\langle x_0, x_1\rangle$.
Since the above encoding of monomials is bijective, the following
claim clearly holds.

\begin{claim}
The bivariate noncommutative polynomial $f(x_0,x_1)$ is nonzero if and
only if the original polynomial $f\in\F\langle
z_1,z_2,\ldots,z_n\rangle$ is nonzero. 
\end{claim}

The degree $D$ of $f(x_0,x_1)$ is clearly bounded by $(n+2)d$.

\begin{definition}
Let $\M \subseteq \{x_0,x_1\}^D$ be a finite set of degree $D$
monomials over variables $\{x_0,x_1\}$. A subset of indices
$I\subseteq [D]$ is said to be an \emph{isolating index set} for $\M$
if there is a monomial $m\in\M$ such that for each $m'\ne m, m'\in\M$,
there is some index $i\in I$ for which $m[i]\ne m'[i]$. I.e.\ no other
monomial in $\M$ agrees with monomial $m$ on all positions in the
index set $I$.
\end{definition}

The following lemma says that every subset of monomials $\M \subseteq
\{x_0,x_1\}^D$ has an isolating index set of size $\log |\M|$. The
proof is a simple halving argument. 

\begin{lemma}\label{isolating-lemma}
Let $\M \subseteq \{x_0,x_1\}^D$ be a finite set of degree $D$
monomials over variables $\{x_0,x_1\}$. Then $\M$ has an isolating
index set of size $k$ which is bounded by $\log |\M|$.
\end{lemma}

\begin{proof}
The monomials $m\in \M$ are seen as indexed from left to right, where
$m[i]$ denotes the variable in the $i^{th}$ position of $m$. Let
$i_1\le D$ be the first index such that not all monomials agree on the
$i^{th}$ position. Let
\begin{eqnarray*}
S_0 & = & \{ m :  m[i_1] = x_0\}\\ 
S_1 & = & \{ m :  m[i_1] = x_1\}.  
\end{eqnarray*}

Either $|S_0|$ or $|S_1|$ is of size at most $|\M|/2$.  Let $S_{b_1}$
denote that subset, $b_1\in\{0,1\}$.  We replace the monomial set $\M$
by $S_{b_1}$ and repeat the same argument for at most $\log |\M|$
steps. Clearly, by this process we identify a set of indices $I=\{i_1,
\ldots, i_k\}$, $k\le \log |\M|$ such that the set shrinks to a
singleton set $\{m\}$. Clearly, $I$ is an isolating index set as
witnessed by the \emph{isolated monomial} $m$.
\end{proof}

\begin{remark}
Notice that the size of the isolating index set denoted $k$ is bounded
by $\log t$ as well as the degree $D$ of the polynomial $f(x_0,x_1)$.
\end{remark}

\subsubsection*{NFA construction}

In our earlier paper \cite{AMS08} (for sparse polynomial identity
testing) we used a deterministic finite state automaton to isolate a
monomial by designing an automaton which accepts a unique
monomial. This will not work for the proof of Theorem~\ref{thm-main}
because the number of states that such a deterministic automaton
requires is the length of the monomial which could be exponentially
large. It turns out that we can use a small \emph{nondeterministic}
finite automaton which will guess the isolating index set for the set
of nonzero monomials of $f$. The complication is that there are
exponentially many wrong guesses. However, it turns out that if we
make our NFA a \emph{substitution automaton}, we can ensure that the
monomials computed on different nondeterministic paths (which
correspond to different guesses of the isolating index set) all have
disjoint support. Once we have this property, it is easy to argue that
for the correct nondeterministic path, the computed commutative
polynomial is nonvanishing (because the isolated monomial cannot be
cancelled). With this intuition, we proceed with the simple technical
details.

We describe the construction of a substitution NFA that substitutes,
on its transition edges, a new commuting variable for the variable
($x_0$ or $x_1$) that it reads. Formally, let $A$ denote the NFA given
by a $5$-tuple $A=\langle Q, \Sigma =
\{x_0,x_1\},\delta,q_o,q_f\rangle$, where $Q = \{q_0,q_1, q_2, \ldots,
q_{k}\}$ and $q_f = q_{k}$. We use the indices $i_1, \ldots,
i_{k}$ from Lemma \ref{isolating-lemma} to define the transition
of $A$.  The set of indices partition each monomial $m$ into $k +
1$ blocks as follows.

\[
m[1,i_1-1] m[i_1] m[i_1 + 1, i_2-1] m[i_2]\cdots\cdots m[i_{k -1}
  + 1, i_{k - 1}] m[i_{k}] m[i_{k + 1}, D],
\] 

where $m[i]$ denotes the variable in $i^{th}$ position of $m$ and
$m[i,j]$ denotes the submonomial of $m$ from positions $i$ to $j$.

We use a new set of variables for different blocks and the indices
$i_1,\ldots, i_{k}$ as follows.  The \emph{block variables} are
$\bigcup_{j\in [k + 1]} \{\xi_j\}$, and the
\emph{index variables} are $\bigcup_{j\in [k]}\{y_{0,j},
y_{1,j}\}$. 

Now we are ready to describe the transitions of the automaton. When
the NFA is reading the input variables in block $j$, it will replace
each $x_b, b\in\{0,1\}$ by block variable $\xi_j$. Then the NFA
nondeterministically decides if block $j$ is over and the current
location is an index in the isolating set. In that case, the NFA
replaces the variable $x_b$ that is read by the index variable
$y_{b,j}$ and the NFA also increments the block number to $j+1$. It
will now make its transitions in the $(j+1)^{st}$ block as described
above.

The NFA is formally described by the following simple transition
rules. For $0\leq i\leq k - 1$, and $b\in\{0,1\}$,

\[
\delta(q_i, x_b)\xrightarrow{\xi_{i+1}} q_i
\]
\[
\delta(q_i, x_b)\xrightarrow{y_{b,i+1}} q_{i+1}.
\]

We depict the description of the automaton in the following figure.

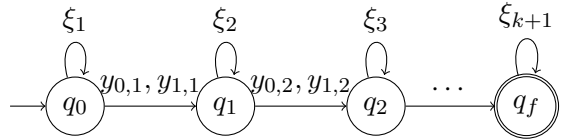
\begin{center}
\begin{tikzpicture}
\node(pseudo) at (-1,0){};
\node(0) at (0,0)[shape=circle,draw]        {$q_0$};
\node(1) at (2,0)[shape=circle,draw]        {$q_1$};
\node(2) at (4,0)[shape=circle,draw]        {$q_2$};
\node(3) at (6,0)[shape=circle,draw,double] {$q_f$};
\path [->]
  (0)      edge                 node [above]  {$y_{0,1}, y_{1,1}$}     (1)
  (1)      edge                 node [above]  {$y_{0,2}, y_{1,2}$}     (2)
  (2)      edge                 node [above]  {$\cdots$}     (3)
  (0)      edge [loop above]    node [above]  {$\xi_1$}     ()
  (1)      edge [loop above]    node [above]  {$\xi_2$}     ()
(2)      edge [loop above]    node [above]  {$\xi_3$}     ()
  (3)      edge [loop above]    node [above]  {$\xi_{k +1}$}   ()
  (pseudo) edge                                       (0);
\end{tikzpicture}
\end{center}

Clearly, the transitions of the automaton $A$ can be described by two
$(k + 1)\times (k + 1)$ adjacency matrices $M_{x_0}$ and
$M_{x_1}$ corresponding to the moves of the automaton on input $x_0$
and input $x_1$. 

More precisely, for variable $x_0$, we take the adjacency matrix
$M_{x_0}$ of the following labeled directed graph extracted from the
above automaton.

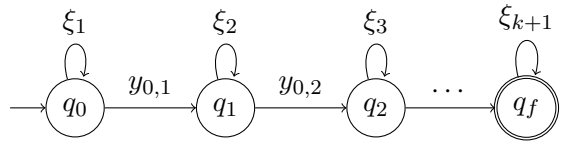
\begin{center}
\begin{tikzpicture}
\node(pseudo) at (-1,0){};
\node(0) at (0,0)[shape=circle,draw]        {$q_0$};
\node(1) at (2,0)[shape=circle,draw]        {$q_1$};
\node(2) at (4,0)[shape=circle,draw]        {$q_2$};
\node(3) at (6,0)[shape=circle,draw,double] {$q_f$};
\path [->]
  (0)      edge                 node [above]  {$y_{0,1}$}     (1)
  (1)      edge                 node [above]  {$y_{0,2}$}     (2)
  (2)      edge                 node [above]  {$\cdots$}     (3)
  (0)      edge [loop above]    node [above]  {$\xi_1$}     ()
  (1)      edge [loop above]    node [above]  {$\xi_2$}     ()
(2)      edge [loop above]    node [above]  {$\xi_3$}     ()
  (3)      edge [loop above]    node [above]  {$\xi_{k+1}$}   ()
  (pseudo) edge                                       (0);
\end{tikzpicture}
\end{center}
The corresponding matrix $M_{x_0}$ of dimension $(k+1)\times(k+1)$, we substitute for $x_0$ is the following.
\begin{displaymath}
\mathbf{M_{x_0}} =
\left( \begin{array}{cccccc}
\xi_1 & y_{0,1} & 0 &\ldots &0 &0\\
0 & \xi_2 & y_{0,2}& \ldots & 0 &0\\
0 & 0 & \xi_3& \ldots & 0 &0\\
\vdots & \vdots & \vdots& \ddots& \vdots& \vdots\\
0 & 0 & 0& \ldots & \xi_{k} &y_{0,k} \\
0 & 0 & 0& \ldots & 0 & \xi_{k+1}
\end{array} \right)
\end{displaymath}
Similarly, for variable $x_1$ we take the adjacency matrix $M_{x_1}$
of the following labeled directed graph.

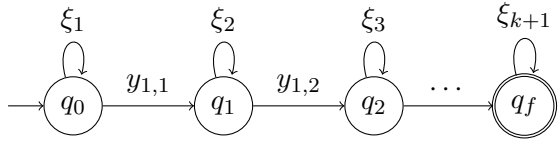
\begin{center}
\begin{tikzpicture}
\node(pseudo) at (-1,0){};
\node(0) at (0,0)[shape=circle,draw]        {$q_0$};
\node(1) at (2,0)[shape=circle,draw]        {$q_1$};
\node(2) at (4,0)[shape=circle,draw]        {$q_2$};
\node(3) at (6,0)[shape=circle,draw,double] {$q_f$};
\path [->]
  (0)      edge                 node [above]  {$y_{1,1}$}     (1)
  (1)      edge                 node [above]  {$y_{1,2}$}     (2)
  (2)      edge                 node [above]  {$\cdots$}     (3)
  (0)      edge [loop above]    node [above]  {$\xi_1$}     ()
  (1)      edge [loop above]    node [above]  {$\xi_2$}     ()
(2)      edge [loop above]    node [above]  {$\xi_3$}     ()
  (3)      edge [loop above]    node [above]  {$\xi_{k +1}$}   ()
  (pseudo) edge                                       (0);
\end{tikzpicture}
\end{center}
The corresponding matrix $M_{x_1}$ of dimension $(k+1)\times(k+1)$, we substitute for $x_1$ is the following.
\begin{displaymath}
\mathbf{M_{x_1}} =
\left( \begin{array}{cccccc}
\xi_1 & y_{1,1} & 0 &\ldots &0 &0\\
0 & \xi_2 & y_{1,2}& \ldots & 0 &0\\
0 & 0 & \xi_3& \ldots & 0 &0\\
\vdots & \vdots & \vdots& \ddots& \vdots& \vdots\\
0 & 0 & 0& \ldots & \xi_{k} &y_{1,k} \\
0 & 0 & 0& \ldots & 0 & \xi_{k+1}
\end{array} \right)
\end{displaymath}

The rows and the columns of the matrices $M_{x_0}$ and $M_{x_1}$ are
indexed by the states of the automaton and the entries are either
block variables or index variables as indicated in the transition
diagram. Let 
\[
f=\sum_{i=1}^t c_i w_i. 
\]

Define the matrix $w_i(M_{x_0},M_{x_1})$ obtained by substituting in
$w_i$ the matrix $M_{x_b}$ for $x_b, b\in\{0,1\}$ and multiplying
these matrices. The following proposition is immediate as $f$ is a
linear combination of the $w_i$'s.

\begin{proposition}\label{prop}
$M_f=f(M_{x_0}, M_{x_1})
 = \sum_{i=1}^t c_iw_i(M_{x_0}, M_{x_1})$. 
\end{proposition}

Now we are ready to prove Theorem \ref{thm-main}. 

\begin{proof}[Proof of Theorem \ref{thm-main}]

We assume that $n$-variate polynomial $f$ is converted to the
bivariate polynomial $f(x_0,x_1)$ over $x_0$ and $x_1$. Let $\M$
denote the set of nonzero monomials of degree $D$ occurring in $f$,
where $D$ is the degree of $f$. Then we can write the polynomial
$f=\sum_{j=1}^t c_j w_j$ in two parts
\[
f= \sum_{w_j\in\M} c_j w_j + \sum_{w_j\not\in\M} c_j w_j,
\]

where $\sum_{w_j\in\M} c_j w_j$ is the homogeneous degree $D$ part of
$f$.

Let us assume, without loss of generality, that $w_1$ is in $\M$ and
it is the monomial isolated in Lemma \ref{isolating-lemma}, and the
isolating index set be $I=\{i_1,i_2,\cdots,i_{k}\}$ such that for
all $w_j\in\M$, $w_j|_I \neq w_1|_I$ (i.e.\ the projections of each
$w_j, j\ne 1$ on index set $I$ differs from the projection of
$w_1$). Let

\[
w_1=x_{b_1}x_{b_2}\cdots x_{b_D},
\] 
where $b_j \in \{0,1\}$.

The following claim is immediate.

\begin{claim}
For each index set $J=\{j_1,j_2,\cdots,j_{k}\}$
nondeterministically picked by the substitution NFA, each nonzero
degree $D$ monomial $w_j$ occurring in $f$ is transformed into a
unique degree $D$ monomial $w_{j,J}$ (which is over the block and
index variables). More precisely, let
$\xi_J=\xi_1^{j_1-1}\xi_2^{j_2-j_1}\cdots\xi^{D-j_{k}}_{k +
  1}$ and $y_{j,J}=y_{a_1,1}y_{a_2,2}\cdots y_{a_{k},k}$. Then
\[
w_{j,J}= \xi_J y_{j,J}.
\]
\end{claim}

Notice that for two distinct index sets $J$ and $J'$ we clearly have
$\xi_J\ne \xi_{J'}$. We also note that $y_{j,J}$ is essentially the
projection of the degree $D$ monomial $w_j$ to the index set $J$; if
variables $x_b$ occurs in the $j_k^{th}$ position of $w_j$ then it is
replaced by $y_{b,k}$ in $y_{j,J}$.

Furthermore, for each $w_j\in\M$ we note that the $(q_o,q_f)^{th}$
entry of the matrix $w_j(M_{x_0},M_{x_1})$ is the sum $\sum_J
w_{j,J}=\sum_J \xi_J y_{j,J}$.  For different index sets $J$ the
monomials $\xi_J y_{j,J}$ are all distinct.

Let $f_J$ be the polynomial
\[
f_J = \sum_{j=1}^t c_j w_{j,J} = \sum_{w_j\in\M} c_j w_{j,J} +
\sum_{w_j\notin\M} c_j w_{j,J}.
\]

\begin{claim}
After the matrix substitution $x_0=M_{x_0}$ and $x_1=M_{x_1}$ in the
polynomial $f$ we note that the $(q_0,q_f)^{th}$ entry of the matrix
$f(M_{x_0},M_{x_1})$ is $\sum_J f_J$.
\end{claim}

The above claim clearly holds because the polynomial $f_J$ is the
contribution of the nondeterministic path corresponding to index set
$J$.

\begin{claim}
For any two index sets $J, J'$ and any monomial $w_j\in\M$, the
corresponding commutative monomials $w_{j,J}$ and $w_{j,J'}$ are
distinct.
\end{claim}

To see this claim it suffices to see that $w_{j,J}=\xi_J y_{j,J}$ and
$w_{j,J'}=\xi_{J'}y_{j,J'}$ and we have already observed that
$\xi_J\ne \xi_{J'}$.

Finally, we focus on the monomial $w_{1,I}$ occurring in the
polynomial $\sum_J f_J$, where $w_1$ is the isolated monomial and $I$
is the isolated index set.

\begin{claim}
The coefficient of $w_{1,I}$ in the polynomial $\sum_J f_J$ is $c_1$.
As a consequence, the polynomial $\sum_J f_J$ which occurs in the
$(q_0, q_f)$ entry of the matrix $M_f = f(M_{x_0}, M_{x_1})$ is
nonzero because the coefficient of $w_{1,I}$ in it is nonzero.
\end{claim}

To see the above claim we note the following points:
\begin{enumerate}
\item For the monomials $w_j\notin \M$ notice that for each index set
$J$ of size $k$, the contribution to the $(q_0,q_f)^{th}$ entry of
the matrix $f(M_{x_0},M_{x_1})$ is a monomial of degree $\deg(w_j)$,
and $\deg(w_j)<D$. Hence, these monomials have no influence on the
coefficient of $w_{1,I}$ in the polynomial $\sum_J f_J$.

\item Consider monomials $w_j\in\M$ for $j\ne 1$. Notice that

\[
w_{1,I}= \xi_I y_{1,I},
\]

and for $j\ne 1$

\[
w_{j,I}= \xi_I y_{j,I},
\]

and the monomials $y_{1,I}$ and $y_{j,I}$ are different because $I$ is
an isolating index set and the monomial $w_1$ is isolated. I.e.\ the
monomials $w_{1,I}$ and $w_{j,I}$ will necessarily differ in the index
variables occurring in them as a consequence of the isolation
property.

Therefore, the monomials $w_j\in\M$ for $j\ne 1$ also have no influence
on the coefficient of $w_{1,I}$ in the polynomial $\sum_J f_J$.
\end{enumerate}

Hence, we conclude that the $(q_0, q_f)^{th}$ entry of the matrix $M_f =
f(M_{x_0}, M_{x_1})$ is a nonzero polynomial $\sum_J f_J$ in the
commuting variables $\bigcup_{j\in [k + 1]} \{\xi_j\}$
and $\bigcup_{j\in [k]}\{y_{0,j}, y_{1,j}\}$. Moreover the degree
of polynomial $\sum_J f_J$ is $D$. Now we can apply
Schwartz-Zippel-Lipton-DeMello Lemma \cite{Sch80,Zippel79,DL78} to
conclude that the polynomial $\sum_J f_J$ will be nonzero over a
suitable extension of size more than $(n+2)d$ of the field $\F$. 
Since the polynomial $f$ is nonzero over the algebra $\mathbb{M}_{k+1}(\F)$, it is also nonzero over the algebra $\mathbb{M}_{\log t+1}(\F)$.
This completes the proof of Theorem~\ref{thm-main}.
\end{proof}

\section*{Acknowledgement}

The number of monomials in a noncommutative polynomial which has an
arithmetic circuit of size $s$ can actually be doubly exponential in
$s$. We thank Srikanth Srinivasan for pointing this out to us.



\end{document}